\documentclass[11pt,a4paper]{article}
\usepackage{amsmath, amsfonts, amssymb}
\usepackage{a4wide}
\usepackage{parskip}
\usepackage{enumitem}
\usepackage{xcolor}

\usepackage{hyperref}
\usepackage{booktabs}
\usepackage{caption}
\usepackage{siunitx}
\usepackage{tabularx}
\usepackage{comment}
\usepackage{multirow}
\usepackage{adjustbox}
\usepackage{mathtools}
\usepackage{caption}
\usepackage{parskip}
\usepackage{graphicx}
\usepackage{adjustbox}
\usepackage{multirow}
\usepackage{amsthm}
\usepackage{bm, makecell}
\usepackage{lscape}
\usepackage{rotating}
\usepackage{floatrow}
\usepackage[noadjust]{cite}

\newtheorem{theorem}{Theorem}[section]
\newtheorem{lemma}[theorem]{Lemma}
\newtheorem{corollary}[theorem]{Corollary}
\newtheorem{proposition}[theorem]{Proposition}
\theoremstyle{definition}
\newtheorem{definition}[theorem]{Definition}

\newtheorem{example}[theorem]{Example}
\newtheorem{remark}[theorem]{Remark}

\numberwithin{equation}{section}

\newcommand{\Tor}{\textnormal{Tor}}

\usepackage{tikz,xcolor,hyperref}
\usepackage{mathdots}

\definecolor{lime}{HTML}{A6CE39}
\DeclareRobustCommand{\orcidicon}{%
	\begin{tikzpicture}
		\draw[lime, fill=lime] (0,0) 
		circle [radius=0.16] 
		node[white] {{\fontfamily{qag}\selectfont \tiny ID}};
		\draw[white, fill=white] (-0.0625,0.095) 
		circle [radius=0.007];
	\end{tikzpicture}
	\hspace{-2mm}
}

\foreach \x in {A, ..., Z}{%
	\expandafter\xdef\csname orcid\x\endcsname{\noexpand\href{https://orcid.org/\csname orcidauthor\x\endcsname}{\noexpand\orcidicon}}
}
\begin{document}
	\date{}
    \title{ Skew Polycyclic Codes over $\frac{\mathbb{F}_{p^m}[u]}{\langle u^t \rangle}$}
		\author{{\bf Akanksha Tiwari\footnote{email: {\tt akankshafzd8@gmail.com}}\orcidA{}, \bf Ritumoni Sarma\footnote{	email: {\tt ritumoni407@gmail.com}}\orcidC{}} \\ Department of Mathematics\\ Indian Institute of Technology Delhi\\Hauz Khas, New Delhi-110016, India}
\maketitle

\begin{abstract}
Let $R^t$ denote the finite chain ring $\frac{\mathbb{F}_{p^m}[u]}{\langle u^t \rangle},$ where $p$ is a prime and $t$ is a positive integer.  In this article, for a prime $p$ and an automorphism $\theta$ of $\mathbb{F}_{p^m}$, we give the structure of the left ideals of the ring $\frac{R^t[x,\Theta]}{\langle f(x) \rangle},$ where $f(x)$ is in the center of the skew polynomial ring $R^t[x,\Theta]$ and $\Theta$ is an automorphism of $R^t$ that extends $\theta$ with $\Theta(u)=u$. These left ideals are also referred to as skew polycyclic codes associated to $f(x).$
In particular, when the central element \( f(x)\) is \(x^{np^s}-\lambda \), where $\lambda=\lambda_0+u\lambda_1+\cdots +u^{t-1}\lambda_{t-1}$ with $\lambda_0\ne0,$ and \( n=1,2 \), we give a more refined form of the left ideals (which are also called skew constacyclic codes). Moreover, the case $\lambda_1 \neq 0$ is analyzed in detail, yielding a simpler form of generators that reveals a more refined structural characterization of the left ideals. As an application, for $n=1,t=3$ and $n=2,t=2$ we give a full description of the left ideals by including certain necessary conditions that were omitted in available literature, preventing the different classes of left ideals from being mutually disjoint and in certain cases, we also compute $i$-th torsion codes.
\medskip

			
\medskip
			

\end{abstract}

\section{Introduction}\label{Section 1}\label{sec1}
Linear codes play a central role in ensuring reliable communication by enabling error detection and error correction during data transmission. Among the broad class of linear codes, cyclic codes occupy a distinguished position due to their algebraic and computational advantages. A cyclic code is invariant under cyclic shifts of its codewords, a property that allows such codes to be represented as ideals in suitable polynomial quotient rings. This correspondence enables the application of algebraic tools for their analysis and facilitates efficient encoding and decoding algorithms. The structural simplicity and strong theoretical framework of cyclic codes have motivated numerous generalizations, including constacyclic, negacyclic, and skew cyclic codes. Codes over finite rings have since emerged as an important area of study, offering new structural insights and enabling constructions that have not yet been achieved by considering fields alone. In particular, linear codes over rings such as $\mathbb{Z}_4$ revealed deep connections with well-known families of nonlinear binary codes via Gray maps studied by Hammons et al. in \cite{hammons1994z}. Beyond their theoretical appeal, ring-based codes provide greater flexibility in algebraic structure.
\par
Constacyclic codes, a generalization of cyclic codes, over finite fields, have been widely investigated. Their structure depends significantly on whether the characteristic of the field divides the code length. If the characteristic does not divide the length, the corresponding defining polynomial has simple roots, whereas if it divides the length, repeated roots occur. Repeated-root constacyclic codes over finite fields were first studied by Castagnoli et al. (\cite{castagnoli1991repeated}) and van Lint (\cite{van1991repeated}). The investigation was later extended to constacyclic codes over finite chain rings. Constacyclic codes over a finite chain ring $R$ have been studied mainly in two directions: when the characteristic of the residue field $\bar{R}$ of $R$ divides the length of the code, and the other when the characteristic of $\bar{R}$ does not divide the length of the code. 
Constacyclic codes over a finite chain ring  $R$ have been studied mainly in two directions: when the characteristic of the residue field $\bar{R}$ of 
$R$ divides the length of the code, and when it does not. Dinh has done several studies on repeated root constacyclic codes, for instance, see \cite{dinh2010constacyclic}-\cite{dinh2020hamming}. In 2016, Liu et al. (\cite{liu2016repeated}) studied repeated-root constacyclic codes of length $3lp^s$ and their dual codes. 

\par
Motivated by these developments, the study was further extended to polycyclic codes, a natural generalization of constacyclic codes. The structure of polycyclic codes over finite chain rings has been examined by several authors; see, for example, \cite{boudine2023polycyclic}, \cite{AlexandreFotue-Tabue2020AdvancesinMathematicsofCommunications}, \cite{boudine2023classification}, \cite{lopez2013polycyclic}, and \cite{AkaKanSar}.
\par
At the same time, the introduction of skew polynomial rings has led to further generalizations of cyclic codes to skew cyclic codes. The non-commutative nature of skew polynomial rings, governed by ring automorphisms, gives rise to skew cyclic and skew constacyclic codes first studied by Boulagouaz and Leroy in \cite{Boulagouaz}. Further, Boucher and Ulmer in \cite{boucher2014linear} investigated linear codes by using a skew polynomial ring with an automorphism as well as a derivation. The interplay between repeated-root structures and skew polynomial settings has attracted growing attention in recent years.
\par
Jitman et al. (\cite{Jitman}) studied the structure and properties of skew constacyclic codes over finite chain rings and, in particular, they studied cyclic codes of length $n$ over $\mathbb{F}_{p^m} +u\mathbb{F}_{p^m}$. Bagheri et al. in \cite{bagheri2022skew} classified skew cyclic codes of length $p^s$ over $\mathbb{F}_{p^m} +u\mathbb{F}_{p^m}$. In \cite{hesari2021self}, Hesari et al. determined the structure of the Euclidean dual of skew cyclic codes of length $p^s$ over $\mathbb{F}_{p^m} +u\mathbb{F}_{p^m}$ in some special cases and established that they are self-dual. Inchaisri et al. in  \cite{inchaisri2022negacyclic} studied negacyclic codes of prime power length over the finite non-commutative chain ring $\mathbb{F}_{p^m}[u,\theta]/\langle u^2\rangle$. Phuto and Klin-Eam in \cite{phuto2022duality} studied the Euclidean dual of skew constacyclic codes of length $p^s.$ Pathak et al. in \cite{pathak2025skew} studied skew nega cyclic codes of length $4p^s$ over $\mathbb{F}_{p^m} +u\mathbb{F}_{p^m}$.
Hesari et al. in \cite{hesari2023skew} studied skew constacyclic codes of length $p^s$ and $2p^s$ over $\mathbb{F}_{p^m} +u\mathbb{F}_{p^m}$ and found their torsion codes. Recently, Hesari and Samei in \cite{hesari2025skew} studied skew constacyclic codes of length $p^s$ over $\mathbb{F}_{p^m} +u\mathbb{F}_{p^m}+u^2\mathbb{F}_{p^m}$ and found their torsion codes.
\par 
To the best of our knowledge, the existing literature, has investigated the structure of skew constacyclic codes of lengths $p^s$, $2p^s$, and $4p^s$ over the finite chain rings $\frac{\mathbb{F}_{p^m}[u]}{\langle u^2 \rangle}$ and $\frac{\mathbb{F}_{p^m}[u]}{\langle u^3 \rangle}$. However, in the description of left ideals in these references, certain necessary conditions were omitted, preventing the different classes of left ideals from being mutually disjoint. In this article, we extend these earlier works from skew constacyclic codes to the more general class of skew polycyclic codes over $R^t:=\frac{\mathbb{F}_{p^m}[u]}{\langle u^t \rangle}$. Moreover, we incorporate the missing necessary conditions required for these families to be pairwise disjoint in several existing descriptions of such codes available in the literature (see, \cite{hesari2023skew}, \cite{hesari2025skew}, and \cite{pathak2025skew}).
\par
 The structure of the article is as follows. In Section \ref{section2}, we fix notation and provide basic definitions. In Section \ref{sec3}, we present the main results of the article. Specifically, for an automorphism $\theta$ of $\mathbb{F}_{p^m}$, we study skew polycyclic codes arising from the quotient ring $\frac{R^t[x,\Theta]}
{\langle f(x) \rangle}$,
where $f(x)$ is in the center of $ R^t[x,\Theta]$, $\Theta$ is a ring automorphism of $R^t$ such that $\Theta(a_0+ua_1+\dots+u^{t-1}a_{t-1})=\theta(a_0)+u\theta(a_1)+\dots + u^{t-1}\theta(a_{t-1})$, which is equivalent to describing the structure of its left ideals. In Section \ref{sec4}, we consider the special case when the central element $f(x)=x^{p^s}-\lambda,$ where
$\lambda=\lambda_0 + u\lambda_1 + \cdots + u^{t-1}\lambda_{t-1} \in 
R^t$ and give a more refined form of left ideals of $\frac{R^t[x,\Theta]}{\langle x^{p^s}-\lambda \rangle}$. In Section \ref{sec5}, we consider the particular case when the central element $f(x)=x^{2p^s}-\lambda,$ where
$\lambda=\lambda_0 + u\lambda_1 + \cdots + u^{t-1}\lambda_{t-1} \in 
R^t$. 
Furthermore, as an application, in Section \ref{sec6}, we completely characterize the left ideals of 
$
\frac{R^3[x,\Theta]}{\langle x^{p^s}-\lambda \rangle}\text{ and }
\frac{R^2[x,\Theta]}{\langle x^{2p^s}-\lambda \rangle},$
by incorporating the missing conditions omitted in the earlier descriptions available in the literature. In certain cases, we also determine the $i$-th torsion of these ideals. Moreover, we provide a counterexample demonstrating that, in the absence of these conditions, the corresponding families are not necessarily disjoint.
In Section \ref{sec7}, we conclude the article.
\section{Notation and Preliminaries}\label{section2}
For a prime number $p$, $\mathbb{F}_{p^m}$ denotes the finite field with $p^m$ elements. Let $R$ denote a ring with unity throughout this section.
\medskip
\begin{definition}
Let $\sigma$ be an automorphism of the ring $R$.  
The set $R[x,\sigma]$ consists of all polynomials over $R$ in indeterminate $x$, where addition is as in the usual polynomial ring over $R$ but multiplication in $R[x,\sigma]$ is determined by the rule
\[
    x a = \sigma(a)\, x \qquad \forall a \in R,
\]
and is extended to all elements of $R[x,\sigma]$ by distributivity and 
associativity.  
With these operations, $R[x,\sigma]$ forms a ring called the 
\emph{skew polynomial ring} over $R$.  
The ring $R[x,\sigma]$ is non-commutative unless $\sigma$ is the identity 
automorphism on $R$.
Moreover, let $R$ be a ring with unity, then the left (as well as right) division algorithm holds in the skew polynomial ring $R[x,\sigma],$ where $\sigma$ is an automorphism of $R$.
\end{definition}
\begin{theorem}[Right Division Algorithm]\cite{jian2022}
    Let $f(x),g(x)\in R[x,\sigma],$ where the leading coefficient of $g(x)$ is a unit. Then there exists $q(x), r(x)\in R[x,\sigma]$ such that
    $$f(x)=q(x)g(x)+r(x),$$
    where $r(x)=0$ or $\deg{r}(x)<\deg{g}(x).$
\end{theorem}
A ring $R$ with unity is a \emph{principal left ideal ring} if every left ideal is 
principally generated. A ring $R$ is called a \emph{local ring} if $R$ has a unique maximal 
right (left) ideal. Furthermore, a ring $R$ is called a \emph{left chain ring} if the set 
of all left ideals of $R$ is linearly ordered under set-theoretic inclusion.
For any two left ideals $I$ and $J$ of a ring $R$, we use $(I:_RJ)$ to denote their ideal quotient, that is, 
$$(I:_RJ)=\{x\in R : xJ \subset I \}.$$
For a finite ring $R$ with unity, an automorphism $\sigma$ of $R$, and $f(x)$ in the center of $R[x,\sigma]$,  left or right ideals of the ring $\frac{R[x,\sigma]}{\langle f(x)\rangle }$ are called \emph{skew polycyclic codes} associated to $f$ over $R.$ If $\alpha \in R$ is a central unit fixed by $\sigma$ and $|\sigma| \mid n$, then any left or right ideal $C$ of the quotient ring $\frac{R[x,\sigma]}{\langle x^n-\alpha\rangle}$ is called an \emph{$(\alpha,\sigma)$-constacyclic code} of length $n$ over $R$. In particular, when $\alpha = 1$, $C$ is called a \emph{skew cyclic code}, and when $\alpha = -1$, it is called a \emph{skew negacyclic code}.

Let $\theta$ be an automorphism of $\mathbb{F}_{p^m}$ and $\mathbb{F}_{p^m}[x,\theta]$ be the skew polynomial ring.
We say that $f(x)$ is a \emph{right divisor} (\emph{left divisor}) of $g(x)$ in $\mathbb{F}_{p^m}[x,\theta]$ and we write 
$f(x)\mid_r g(x)$ (resp. $f(x)\mid_l g(x)$) if there exists a skew polynomial $h(x)$ such that 
$g(x)=h(x)f(x)$ (resp. $g(x)=f(x)h(x)$).

\begin{definition}
    Suppose $f(x), g(x)$ are skew polynomials in $\mathbb{F}_{p^m}[x,\theta]$.  
The \emph{greatest common right divisor} of $f(x)$ and $g(x)$ is the monic polynomial 
$d(x)\in \mathbb{F}_{p^m}[x,\theta]$, where $d(x)\mid_r f(x)$, $d(x)\mid_r g(x)$ and for any 
$e(x)\in \mathbb{F}_{p^m}[x,\theta]$ such that $e(x)\mid_r f(x)$ and $e(x)\mid_r g(x)$, 
we have $e(x)\mid_r d(x)$.  
We denote $d(x)$ by $\gcd_r(f(x),g(x))$.

\end{definition}

\begin{definition}
The \emph{least common left multiple} ($\operatorname{lcm_l}$) of $f(x)$ and $g(x)$ is the unique monic 
polynomial $m(x)=\operatorname{lcm_l}(f(x),g(x))$ such that 
$f(x)\mid_r m(x)$, $g(x)\mid_r m(x)$ and for any $n(x)\in \mathbb{F}_{q}[x,\theta]$ such that 
$f(x)\mid_r n(x)$ and $g(x)\mid_r n(x)$, we have $m(x)\mid_r n(x)$.
\end{definition}

Let $\frac{\mathbb{F}_{p^m}[u]}{\langle u^t\rangle}$ be denoted by $R^t$, where $t\in \mathbb{N}.$ It is well known that $R^t$ is a chain ring. An arbitrary element $a$ of $R^t$ can be uniquely written as $a_0+a_1u+\dots+a_{t-1}u^{t-1},$ where $a_i\in\mathbb{F}_{p^m}$ for $0\leq i\leq t-1.$ For an automorphism $\theta$ of $\mathbb{F}_{p^m}$, the center of the skew polynomial ring $\mathbb{F}_{p^m}[x,\theta]$ is $\mathbb{F}_{p^m}^{\theta}[x^{|\theta|}]$, where $\mathbb{F}_{p^m}^{\theta}$ is the fixed field of $\theta$ and $|\theta|$ denotes the order of $\theta.$ The map $\Theta:R^t\rightarrow R^t$ defined by $\Theta(a)=\theta (a_0)+\theta(a_1)u+\dots+\theta(a_{t-1})u^{t-1}$ is an automorphism of $R^t$. Obviously,  $|\Theta|=|\theta|$. Let $R^t [x,\Theta]$ be the skew polynomial ring with usual addition and multiplication given by $xr=\Theta(r)x.$ Further assume that $f(x)=f_0+f_1x^{|\Theta|}+\dots+f_rx^{r|\Theta|}\in R^{t}[x,\Theta]$ is a polynomial such that $\Theta(f_i)=f_i$ for $0\le i \le n$. Then $f(x)$ is also in the center of $R^t[x,\Theta]$ and hence a left (or right) ideal generated by $f(x)$ is a two-sided ideal  \cite{mcdonald1974finite}. Let us denote by $R^{t,f}_{\Theta}$ the quotient ring $\frac{R^t[x,\Theta]}{\langle f(x)\rangle}.$ We denote the image of $f(x)$ under modulo $u$ map by $\overline{f(x)}.$ We further assume that $\deg\overline{f(x)}=d.$ In that case $\langle \overline{f(x)}\rangle$ is a two sided ideal in $\mathbb{F}_{p^m}[x,\theta].$
In particular case when $f(x)$ is $x^{np^s}-\lambda$, where $\lambda =\lambda_0+u\lambda_1+\dots+u^{t-1}\lambda_{t-1}\in R^t,$ $\lambda_i\in\mathbb{F}_{p^m}$, for $0\leq i\leq t-1$, and $\lambda_0\ne 0$ with the assumption that $|\Theta|=|\theta|$ divides $np^s$ and $\lambda$ is fixed by $\Theta,$ then $\langle x^{np^s}-\lambda \rangle$ is a two-sided ideal in $R^t [x,\Theta]$ and in this case we denote $\frac{R^t[x,\Theta]}{\langle x^{np^s}-\lambda \rangle}$ by $R^{t,\,x^{np^s}-\lambda}_{\Theta}$. With these notations, we give the following definition.
\begin{definition}
For a left ideal $I$ of $R^{t,f}_{\Theta}$, the $i$-th torsion of $I$, denoted by $\textnormal{Tor}_i(I)$ is given by,$$\textnormal{Tor}_i(I):=\mu \{c(x)\in R^{t,f}_{\Theta}\,:\, c(x)u^{i}\in I \},$$ for $0\leq i \leq t-1$.
\end{definition}
From the definition of the $i$-th torsion of a left ideal $I$ in $R^{t,f}_{\Theta},$ it is easy to see that $\textnormal{Tor}_i(I)$ is a left ideal in $\frac{\mathbb{F}_{p^m}[x,\theta]}{\langle \overline{f(x)} \rangle}$ and $\textnormal{Tor}_i(I)\subset\textnormal{Tor}_{i+1}(I)$ for $0\leq i\leq t-2.$
\section{Skew Polycyclic Codes over \texorpdfstring{$\frac{\mathbb{F}_{p^m}[u]}{\langle u^t\rangle}$}{}}\label{sec3}
With the assumptions and notations given in Section \ref{section2}, we begin by proving a result of independent interest that will play a key role in determining the structure of the left ideals (right ideals) of $R^{t,f}_{\Theta}$, and it is an extension of Proposition 2.1 of \cite{AkaKanSar} in a non-commutative setting.
\begin{proposition}\label{FormofIdeals}
Let $A$ and $B$ be two rings with unity and let $\phi: A \rightarrow B$ be a surjective ring homomorphism with $ ker(\phi)=A \pi.$ If $I$ is a left (respectively, right) ideal of $A$ such that $\phi(I)$ is a finitely generated left (respectively, right) ideal of $B$, then there exist $x_1,x_2,\dots,x_n$ in $I$ such that 
$$I= Ax_1+Ax_2+\dots+Ax_n+(I:_{A}\pi)\pi,$$where $(I:_{A}\pi)=\{a\in A\,|\, a\pi\in I\}.$
\end{proposition}
\begin{proof}We are given that $\phi(I)$ is finitely generated left ideal of $B$. Let us assume that  $\phi(I)= By_1+By_2+\dots+By_n \lhd_{l}B$, where $y_1,y_2,\dots,y_n$ in $\phi(I)$. Then we have $x_1,x_2,\dots,x_n\in I$ such that $\phi(x_i)=y_i$ for $1\leq i\leq n$. Let $x\in I,$ then there exist $z_1,z_2,\dots,z_n\in B$ such that $\phi(x)=z_1y_1+z_2y_2+\dots+z_ny_n$ (as $\phi(I)$ is a left ideal). For $1\leq i\leq n$, there exists $u_i\in A$ such that $\phi(u_i)=z_i$ (as $\phi$ is surjective). Thus $\phi(x)=\phi(u_1)\phi(x_1)+\phi(u_2)\phi(x_2)+\dots+\phi(u_n)\phi(x_n)=\phi(u_1x_1+u_2x_2+\dots+u_nx_n)$. Then $x-\underset{i=1}{\overset{n}{\sum}}u_ix_i\in \ker(\phi)=\langle \pi\rangle$. In fact, $x-\underset{i=1}{\overset{n}{\sum}}u_ix_i\in  ker(\phi)\cap I.$ Hence, $I\subset Ax_1+Ax_2+\dots+Ax_n +ker(\phi)\cap I. $ It is easy to see that $ker(\phi)\cap I=(I:_{A}\pi)\pi,$ where $(I:_{A}\pi)=\{a\in A\,|\, a\pi\in I\}.$\hfill $\square$
\end{proof}
Further, we define a surjective ring homomorphism from $R^{t,f}_{\Theta}$ to $\frac{\mathbb{F}_{p^m}[x,\theta]}{\langle  \overline{f(x)} \rangle}$ in the following lemma.
\begin{lemma}
Let $\mu :R^{t,f}_{\Theta}\rightarrow \frac{\mathbb{F}_{p^m}[x,\theta]}{\langle  \overline{f(x)} \rangle}$ such that $\mu(c(x))=c(x) (mod\, u).$ Then $\mu$ is a surjective ring homomorphism with $\ker (\mu)=\langle u \rangle$.
\end{lemma} 
Consider the set
\[
\mathcal{B}_{\overline{f}}:=\left\{a(x)\in \frac{\mathbb{F}_{p^m}[x,\theta]}{\langle \overline{f(x)} \rangle}\;:\; a(x)\mid_{r}\overline{f(x)}\right\}.
\]
The structure of the left ideals of the ring $\frac{\mathbb{F}_{p^m}[x,\theta]}{\langle \overline{f(x)} \rangle}$ is given by the following theorem.
\begin{theorem}
The ring $\frac{\mathbb{F}_{p^m}[x,\theta]}{\langle\overline{ f(x)}\rangle}$ is a left principal ideal ring. The left ideals are given by $a(x)+\langle \overline{ f(x)}\rangle$, where $a(x)\in\mathcal{B}_{\overline{f}}$.
\end{theorem}
\begin{proof}
The argument of the proof is similar to Proposition 2.6 in  \cite{hesari2025skew}. 
\end{proof}
We now describe the structure of the left ideals of $R^{t,f}_{\Theta}$.
    
\begin{theorem}\label{nmain theorem}
The left ideals of the ring $R^{t,f}_{\Theta}$ have one of the following forms.
\begin{itemize}
\item [(i)] Trivial ideals $\langle 0\rangle,$ $\langle 1 \rangle.$
\item [(ii)] Any generator of a non-trivial left ideal contained in $\langle u\rangle $ has the form:
        $$u^{(t-1)-i}b_i(x)-u^{(t-1)-(i-1)}g_i(x),$$
        where $g_i(x)\in R^{t,f}_{\Theta}$, $b_{i}(x)\in\mathcal{B}_{\overline{f}}, \deg b_i(x)\leq d-1$ for $0\leq i \leq t-2.$ An arbitrary left ideal contained in $\langle u \rangle$ is given as:
\begin{align*}
        R^{t,f}_{\Theta}&(u^{(t-1)-i_1}b_{i_1}(x)-u^{(t-1)-(i_1-1)}g_{i_1}(x))+\dots+\\&R^{t,f}_{\Theta}( u^{(t-1)-i_n}b_{i_n}(x)-u^{(t-1)-(i_n-1)}g_{i_n}(x)),
\end{align*} where $0\leq i_1<i_2<\dots<i_n\leq t-2$, $ \deg b_{i_j}(x)\leq d-1$, and $g_{i_j}(x)\in R^{t,f}_{\Theta}$ for $1\leq j\leq n$, and whenever $u^{(t-1)-i_j}c(x)+u^{(t-1)-(i_j-1)}g(x)\in R^{t,f}_{\Theta} (u^{(t-1)-i_k}b_{i_k}(x)-u^{(t-1)-(i_k-1)}g_{i_k}(x))+\cdots+R^{t,f}_{\Theta} (u^{(t-1)-i_n}b_{i_n}(x)-u^{(t-1)-(i_n-1)}g_{i_n}(x))$ for some $c(x)\in \frac{\mathbb{F}_{p^m}[x,\theta]}{\langle\overline{ f(x)}\rangle}$, $g(x)\in R^{t,f}_{\Theta}$, and $j<k,$ then $b_{i_j}(x)\mid_r c(x).$
\item[(iii)] Any non-trivial ideal not contained in $\langle u \rangle$ has the form:
        $$ R^{t,f}_{\Theta}(b(x)+ ug(x)) +I,$$ where $g(x)\in R^{t,f}_{\Theta}$, $b(x)\in\mathcal{B}_{\overline{f}},\,\deg b(x)\leq d-1$, $I $ is an ideal described as in Type (ii), and whenever $u^{(t-1)-i_j}c(x)+u^{(t-1)-(i_j-1)}g(x)\in R^{t,f}_{\Theta} (u^{(t-1)-i_k}b_{i_k}(x)-u^{(t-1)-(i_k-1)}g_{i_k}(x))+\cdots+R^{t,f}_{\Theta} (u^{(t-1)-i_n}b_{i_n}(x)-u^{(t-1)-(i_n-1)}g_{i_n}(x))$ for some $c(x)\in \frac{\mathbb{F}_{p^m}[x,\theta]}{\langle\overline{ f(x)}\rangle}$,  $g(x)\in R^{t,f}_{\Theta},$ and $j<k,$ then $b_{i_j}(x)\mid_r c(x).$
\end{itemize}
\end{theorem}

    \begin{proof} Let $I$ be a nontrivial ideal of $R^{t,f}_{\Theta}$ contained in $\langle u \rangle$. 
 We will prove the result by induction on $t$. Note that the result trivially holds for $t=1$. Let $t>1$ and assume that the result holds for $t-1$. \\Define two maps, $\Phi:R^{t,f}_{\Theta}\rightarrow R^{t-1, f\,(\textnormal{ mod }u^{t-1})}_{\Theta|_{R^{t-1}}}:=\frac{R^{t-1}[x,\Theta|_{R^{t-1}}]}{\langle f(x)\,(\textnormal{ mod }u^{t-1})\rangle}$ such that $\Phi(c(x))=c(x)(\textnormal{mod } u^{t-1}$) and
 $ \mu :R^{t,f}_{\Theta}\rightarrow \frac{\mathbb{F}_{p^m}[x,\theta]}{\langle\overline{ f(x)}\rangle}$ such that $\mu(c(x))=c(x)(\textnormal{ mod }u).$

Let $I \subset \langle u \rangle$, then $\Phi(I) \subset \langle u \rangle $.

    Case 1. Let $\Phi(I)=\langle 0 \rangle.$
    Then, by Proposition \ref{FormofIdeals}, $I=(I:u^{t-1})u^{t-1}$. Since $I$ is non-trivial, $(I:u^{t-1})\not\subset \langle u\rangle$.  Note that $\mu(I:u^{t-1})$ is a non-trivial ideal of $\frac{\mathbb{F}_{p^m}[x,\theta]}{\langle\overline{ f(x)}\rangle}$. 
    Thus $\mu(I:u^{t-1}) =\frac{\mathbb{F}_{p^m}[x,\theta]}{\langle\overline{ f(x)}\rangle} b_1(x),$ for some $b_1(x)\in\mathcal{B}_{\overline{f}}$. Then, by Proposition \ref{FormofIdeals}, $(I:u^{t-1})=R^{t,f}_{\Theta}(b_1(x)+ug(x))+u(I:u^t),$ for some $g(x)\in R^{t,f}_{\Theta}$. Hence, $I=\langle u^{t-1} b_1(x)\rangle,$ for some $b_1(x)\in\mathcal{B}_{\overline{f}}.$

    Case 2. $\Phi(I)\not=\langle 0 \rangle.$ 

    Since $\Phi(I)\subset \langle u\rangle$, by induction hypothesis
    \begin{align*}
\Phi(I)=R^{t-1, f\,(\textnormal{ mod }u^{t-1})}_{\Theta|_{R^{t-1}}}&\big( u^{(t-2)-i_1}b_{i_1}(x)+u^{(t-2)-(i_1-1)}g_{(i_1-1)}(x)\big)+ R^{t-1, f\,(\textnormal{ mod }u^{t-1})}_{\Theta|_{R^{t-1}}}\big( u^{(t-2)-i_2}\\& b_{i_2}(x)+u^{(t-2)-(i_2-1)}
g_{(i_2-1)}\big)+\dots+ R^{t-1, f\,(\textnormal{ mod }u^{t-1})}_{\Theta|_{R^{t-1}}}\big(u^{(t-2)-i_n}b_{i_n}(x)\\& +u^{(t-2)-(i_n-1)}g_{(i_n-1)}(x)\big), 
\end{align*} where $0\leq i_1<i_2<\dots<i_n\leq t-2$, $b_{i_j}(x)\in \mathcal{B}_{\overline{f}}$ for $1\leq j\leq n$, and $g_{i_j-1}(x)\in R^{t-1, f\,(\textnormal{ mod }u^{t-1})}_{\Theta|_{R^{t-1}}}$ for $1\leq j\leq n.$

Then, by Proposition \ref{FormofIdeals}, we have,
\begin{align*}
I=R^{t,f}_{\Theta}& \big(u^{(t-2)-i_1}b_{i_1}(x)+u^{(t-2)-(i_1-1)}g_{(i_1-1)}(x)+u^{t-1}q_1(x)\big)+R^{t,f}_{\Theta}\big( u^{(t-2)-i_2}b_{i_2}(x)+\\&u^{(t-2)-(i_2-1)}g_{(i_2-1)}(x)+u^{t-1}q_2(x)\big)+\dots+R^{t,f}_{\Theta}\big( u^{(t-2)-i_n}b_{i_n}(x)+u^{(t-2)-(i_n-1)}\\&g_{(i_n-1)}(x)+ u^{t-1}q_n(x)\big) +(I:u^{t-1})u^{t-1},
\end{align*}where  $q_1(x),q_2(x),\dots,q_n(x)\in \frac{\mathbb{F}_{p^m}[x,\theta]}{\langle\overline{ f(x)}\rangle}$, $0\leq i_1<i_2<\dots<i_n\leq t-2$.\\
Hence, 
\begin{align*}
 I=R^{t,f}_{\Theta}&\big(u^{(t-1)-(i_1+1)}b_{i_1}(x)+u^{(t-1)-i_1}g_{i_1-1}(x)+u^{t-1}q_1(x)\big)+ R^{t,f}_{\Theta}\big(u^{(t-1)-(i_2+1)}b_{i_2}(x)+\\&u^{(t-1)-i_2}g_{i_2-1}(x)+u^{t-1}q_2(x)\big)+\dots+R^{t,f}_{\Theta}\big(u^{(t-1)-(i_n+1)}b_{i_n}(x)+u^{(t-1)-i_n}g_{i_n-1}(x)\\&+u^{t-1}q_n(x)\big) +(I:u^{t-1})u^{t-1},   
\end{align*}where $0\leq i_1<i_2<\dots<i_n\leq t-2$, $b_{i_j}(x)\in \mathcal{B}_{\overline{f}}$ for $1\leq j\leq n$, and $g_{i_j-1}(x)\in R^{t-1, f\,(\textnormal{ mod }u^{t-1})}_{\Theta|_{R^{t-1}}}$ for $1\leq j\leq n.$ 
If $(I:u^{t-1})\subset \langle u \rangle,$ then $u^{t-1}(I:u^{t-1})=\langle 0\rangle$ and hence
\begin{align*}
    I=R^{t,f}_{\Theta}&\big(u^{(t-1)-(i_1+1)}b_{i_1}(x)+u^{(t-1)-i_1}g_{i_1-1}(x)+u^{t-1}q_1(x)\big)+ R^{t,f}_{\Theta}\big(u^{(t-1)-(i_2+1)}b_{i_2}(x)+\\&u^{(t-1)-i_2}g_{i_2-1}(x)+u^{t-1}q_2(x)\big)+\dots+R^{t,f}_{\Theta}\big(u^{(t-1)-(i_n+1)}b_{i_n}(x)+u^{(t-1)-i_n}g_{i_n-1}(x)\\&+u^{t-1}q_n(x)\big),
\end{align*} 
where $0\leq i_1<i_2<\dots<i_n\leq t-2$, $b_{i_j}(x)\in \mathcal{B}_{\overline{f}}$ for $1\leq j\leq n$, and $g_{i_j-1}(x)\in R^{t-1, f\,(\textnormal{ mod }u^{t-1})}_{\Theta|_{R^{t-1}}}$ for $1\leq j\leq n.$ 

Also,

$u^{(t-2)-i_j}c(x)+u^{(t-2)-(i_j-1)}g(x)\in R^{t-1, f\,(\textnormal{ mod }u^{t-1})}_{\Theta|_{R^{t-1}}} (u^{(t-2)-i_k}b_{i_k}(x)-u^{(t-2)-(i_k-1)}g_{i_k}(x))+\cdots+ R^{t-1, f\,(\textnormal{ mod }u^{t-1})}_{\Theta|_{R^{t-1}}} (u^{(t-2)-i_n}b_{i_n}(x)-u^{(t-2)-(i_n-1)}g_{i_n}(x)) \\\iff u^{(t-1)-(i_j+1)}c(x)+u^{(t-1)-i_j}g(x)\in R^{t,f}_{\Theta} (u^{(t-1)-(i_k+1)}b_{i_k}(x)-u^{(t-1)-i_k}g_{i_k}(x))+\cdots+R^{t,f}_{\Theta} (u^{(t-1)-(i_n+1)}b_{i_n}(x)-u^{(t-1)-i_n}g_{i_n}(x)) $. 

If $(I:u^{t-1})\not \subset \langle u \rangle,$ then $\mu(I:u^{t-1})$ is a non-zero ideal of $\frac{\mathbb{F}_{p^m}[x,\theta]}{\langle\overline{ f(x)}\rangle}$. Thus $\mu(I:u^{t-1})=\frac{\mathbb{F}_{p^m}[x,\theta]}{\langle\overline{ f(x)}\rangle}\big( b_1(x)\big),$ for some $b_1(x)\in \mathcal{B}_{\overline{f}}$. Then, by Proposition \ref{FormofIdeals}, we have that $(I:u^{t-1})=R^{t,f}_{\Theta}\big( b_1(x)+u q(x)\big) +u(I:u^{t}),$ for some $q(x)\in R^{t,f}_{\Theta}$. Consequently, we have
\begin{align*}
    I=R^{t,f}_{\Theta}&\big(u^{(t-1)-(i_1+1)}b_{i_1}(x)+u^{(t-1)-i_1}g_{i_1-1}(x)+u^{t-1}q_1(x)\big)+R^{t,f}_{\Theta}\big( u^{(t-1)-(i_2+1)}b_{i_2}(x)+\\&u^{(t-1)-i_2}g_{i_2-1}(x)+u^{t-1}q_2(x)\big)+\dots+R^{t,f}_{\Theta}\big(u^{(t-1)-(i_n+1)}b_{i_n}(x)+u^{(t-1)-i_n}g_{i_n-1}(x)\\&+u^{t-1}q_n(x)\big) +R^{t,f}_{\Theta}(u^{t-1}b_1(x)).
\end{align*}
Hence, 
\begin{align*}
    I=R^{t,f}_{\Theta}(&u^{t-1}b_1(x))+R^{t,f}_{\Theta}\big(u^{(t-1)-(i_1+1)}b_{i_1}(x)+u^{(t-1)-i_1}g_{i_1-1}(x)+u^{t-1}q_1(x)\big)+R^{t,f}_{\Theta}\\&\big( u^{(t-1)-(i_2+1)}b_{i_2}(x)+u^{(t-1)-i_2}g_{i_2-1}(x)+u^{t-1}q_2(x)\big)+\dots+R^{t,f}_{\Theta}\big(u^{(t-1)-(i_n+1)}\\&b_{i_n}(x)+u^{(t-1)-i_n}g_{i_n-1}(x)+u^{t-1}q_n(x)\big),
\end{align*}where $0\leq i_1<i_2<\dots<i_n\leq t-2$, $b_{i_j}(x)\in \mathcal{B}_{\overline{f}}$ for $1\leq j\leq n$, and $g_{i_j-1}(x)\in R^{t-1, f\,(\textnormal{ mod }u^{t-1})}_{\Theta|_{R^{t-1}}}$ for $1\leq j\leq n.$ Also, as $\mu(I:u^{t-1})=\frac{\mathbb{F}_{p^m}[x,\theta]}{\langle\overline{ f(x)}\rangle}\big( b_1(x)\big),$ whenever $u^{t-1}c(x)\in R^{t,f}_{\Theta}\big(u^{(t-1)-(i_k+1)}b_{i_k}(x)+u^{(t-1)-i_k}g_{i_k-1}(x)+u^{t-1}q_k(x)\big)+\dots+R^{t,f}_{\Theta}\big(u^{(t-1)-(i_n+1)}b_{i_n}(x)+u^{(t-1)-i_n}g_{i_n-1}(x)+u^{t-1}q_n(x)\big)$ for $k\geq 1,$ for some $c(x)\in \frac{\mathbb{F}_{p^m}[x,\theta]}{\langle\overline{ f(x)}\rangle},$ $b(x)\mid_r c(x)$.

Let $I$ be not contained in $\langle u \rangle$. That means $\mu(I)=\frac{\mathbb{F}_{p^m}[x,\theta]}{\langle\overline{ f(x)}\rangle}\big( b(x)\big)$ for some $b(x)\in\mathcal{B}_{\overline{f}}.$ Then from  Proposition \ref{FormofIdeals}, we have $I=R^{t,f}_{\Theta} (b(x)+ur(x))+ u (I:u) $ for some $r(x)\in R^{t,f}_{\Theta}.$ Note that $u(I:u)$ is an ideal of $R^{t,f}_{\Theta}$ that is contained in $\langle u\rangle.$  Hence $I=R^{t,f}_{\Theta} (b(x)+ur(x)) + J,$ where $J$ is an ideal of $R^{t,f}_{\Theta}$ that is contained in $\langle u\rangle$ described above. Hence, the proof follows. \hfill $\square$
\end{proof}
\section{Skew Constacyclic codes of length \texorpdfstring{$p^s$}{} over \texorpdfstring{$\frac{\mathbb{F}_{p^m}[u]}{\langle u^t\rangle}$}{}}\label{sec4}
With the assumptions and notations in Section \ref{section2}, in this section, for the particular case $f(x)=x^{p^s}-\lambda,$ where $\lambda=\lambda_0+u\lambda_1+\dots+u^{t-1}\lambda_{t-1}$, $\Theta(\lambda)=\lambda$,  $|\Theta|\mid p^s,$ and $\lambda_0\ne 0,$ we give the structure of ideals of $R^{t,x^{p^s}-\lambda}_{\Theta}$. Since $\lambda_0 \in \mathbb{F}_{p^m},$ define $\lambda_{0,0}=\lambda_0^{-p^{m-r}}$, where $s=lm+r$ and $0\leq r\leq m-1.$ Then $\lambda_{0,0}^{p^s}=\lambda_0^{-1}.$
We have the following theorem for the central elements of $R^{t,x^{p^s}-\lambda}_{\Theta}$.
\begin{theorem} [\cite{Jitman}, Proposition 2.3]
Let $g(x),h(x)\in R^{t,x^{p^s}-\lambda}_{\Theta}.$ If $g(x)h(x)$ is in the center of the ring $R^{t,x^{p^s}-\lambda}_{\Theta}$, then $g(x)h(x)=h(x)g(x).$
\end{theorem}
    \begin{lemma}\label{relation between}
        Let $k\geq 1$ be the smallest integer such that $\lambda_k\neq 0.$ Then in $R^{t,x^{p^s}-\lambda}_{\Theta},$ we have $\langle  (\lambda_{0,0}x-1)^{p^s}\rangle =\langle u^k\rangle.$ 
    \end{lemma}
    \begin{proof}
        In  $R^{t,x^{p^s}-\lambda}_{\Theta},$ we have $(\lambda_{0,0}x-1)^{p^s}=\lambda_0^{-1}\lambda -1=u^k(\lambda_0^{-1}\lambda_k+\lambda_0^{-1}\lambda_{k+1}u+\dots+\lambda_0^{-1}\lambda_{t-1}u^{t-1-k}).$ Hence, $\langle (\lambda_{0,0}x-1)^{p^s} \rangle=\langle u^k\rangle.$ \hfill $\square$
    \end{proof}
    Hence $\lambda_{0,0}x-1$ is a nilpotent in the ring $R^{t,x^{p^s}-\lambda}_{\Theta}$ having nilpotency index $\lceil t/k \rceil.$ Due to Theorem 2.2 in Section \ref{section2}, we can give the following representation of every element of $R^{t,x^{p^s}-\lambda}_{\Theta}.$ 
    \begin{remark}\label{Representation}
    \begin{itemize}
        \item[(i)] An arbitrary element $c(x)$ of $R^{t,x^{p^s}-\lambda}_{\Theta}$ can be uniquely written as \\   \begin{equation}\label{form of elt in t}
            c(x)=\underset{i=0}{\overset{p^s-1}{\sum}}a_{0,i}(\lambda_{0,0}x-1)^i+u\underset{i=0}{\overset{p^s-1}{\sum}}a_{1,i}(\lambda_{0,0}x-1)^i+\dots+u^{t-1}\underset{i=0}{\overset{p^s-1}{\sum}}a_{t-1,i}(\lambda_{0,0}x-1)^i
        \end{equation}
        where $a_{j,i}\in \mathbb{F}_{p^m}$ for $0\leq j \leq t-1,\,0\leq i\leq p^s-1.$
    \end{itemize}
    \end{remark}
    It is noteworthy that in a non-commutative ring, the sum of a unit and a nilpotent element may not be a unit.
     Consider the set
\[
\mathcal{B}_{\lambda}^1:=\left\{a(x)\in \frac{\mathbb{F}_{p^m}[x,\theta]}{\langle x^{p^s}-\lambda_0\rangle}\;:\; a(x)\mid_{r}(\lambda_{0,0}x-1)^{p^s}\right\}.
\]
The left ideals of the ring $\frac{\mathbb{F}_{p^m}[x,\theta]}{\langle x^{p^s}-\lambda_0\rangle}$ are characterized by the following theorem.
     \begin{theorem}[\cite{hesari2025skew}, Proposition 2.6]
     $\frac{\mathbb{F}_{p^m}[x,\theta]}{\langle (\lambda_{0,0}x-1)^{p^s}}$ is a left principal ideal ring The left ideals are given by $a(x)+\langle (\lambda_{0,0}x-1)^{p^s}\rangle$, where $a(x)\in\mathcal{B}_{\lambda}^1$.
         
     \end{theorem}
    Similar to the proof of Lemma 3.1 in \cite{bagheri2022skew}, we have the following result.
\begin{lemma}\label{invertible in Fpm[x,theta]/xps-lambda_0}
    Let $g(x)$ be a non-zero polynomial in $\frac{\mathbb{F}_{p^m}[x,\theta]}{\langle x^{p^s}-\lambda_0\rangle}.$ Then $g(x)$ is left invertible if and only if $\textnormal{gcd}_r(g(x),(\lambda_{0,0}x-1)^{p^s})=1.$
\end{lemma}
\begin{proposition}\label{invertible in Fpm[u]/<ut>[x,Theta]/xps-lambda}
    Let $h(x)$ be a non-zero element in $R^{t,x^{p^s}-\lambda}_{\Theta}.$ Then $h(x)$ is left invertible in $R^{t,x^{p^s}-\lambda}_{\Theta}$ if and only if $\mu(h(x))$ is left invertible in $\frac{\mathbb{F}_{p^m}[x,\theta]}{\langle x^{p^s}-\lambda_0\rangle}.$
\end{proposition} 
\begin{proof}The forward part is obvious, so it is enough to prove the converse. Let $h(x)=h_0(x)+uh_1(x)+\dots+u^{t-1}h_{t-1}(x)$. Assume that  $\mu(h(x))=h_0(x)$ is invertible in $\frac{\mathbb{F}_{p^m}[x,\theta]}{\langle x^{p^s}-\lambda_0\rangle}.$  Then $h_0(x)^{-1}h(x)=1+uh_0^{-1}(x)h_1(x)+u^2h_0^{-1}(x)h_2(x)+\dots+u^{t-1}h_0(x)^{-1}h_{t-1}(x)$. Note that $h_0(x)^{-1}h(x)$ is left invertible if there exists $f(x)=1+uf_1(x)+\dots u^{t-1}f_{t-1}(x)$ such that $f(x)h_0(x)^{-1}h(x)=1.$ Furthermore $f(x)h_0(x)^{-1}h(x)=1$ is equivalent to the following system.
\begin{small}
\[
\begin{pmatrix}
1 & 0 &0&\dots&0&0 \\
    h_0(x)^{-1}h_1(x) & 1&0&\dots&0&0\\
h_0(x)^{-1}h_2(x) & h_0(x)^{-1}h_1(x)&1&\dots&0&0\\
\vdots & \vdots&\vdots&\vdots&\vdots&\vdots\\
h_0(x)^{-1}h_{t-2}(x) & h_0(x)^{-1}h_{t-3}(x)&h_0(x)^{-1}h_{t-4}(x)&\dots&h_0(x)^{-1}h_{1}(x)&1
\end{pmatrix}
\begin{pmatrix}
f_1(x)\\
f_2(x)\\
f_3(x)\\
\vdots\\
f_{t-1}(x)
\end{pmatrix}
\]
\[=
\begin{pmatrix}
-h_0(x)^{-1}h_1(x)\\
-h_0(x)^{-1}h_2(x)\\
-h_0(x)^{-1}h_3(x)\\
\vdots\\
-h_0(x)^{-1}h_{t-1}(x)
\end{pmatrix}
\]
\end{small}
Hence, $h_0(x)^{-1}h(x)$ is left invertible as the above system is consistent.\hfill{$\square$}
\end{proof}
 We now describe the structure of the left ideals of the ring $R^{t,x^{p^s}-\lambda}_{\Theta}$, which directly follows from Theorem \ref{nmain theorem}.
    
    \begin{theorem}\label{main theorem}
    The left ideals of the ring $R^{t,x^{p^s}-\lambda}_{\Theta}$ have one of the following forms.
    \begin{itemize}
        \item [(i)] Trivial ideals $\langle 0\rangle,$ $\langle 1 \rangle.$
        \item [(ii)] Any generator of a non-trivial left ideal contained in $\langle u\rangle $ has the form:
        $$u^{(t-1)-i}b_i(x)-u^{(t-1)-(i-1)}g_i(x),$$
        where $g_i(x)\in R^{t,x^{p^s}-\lambda}_{\Theta}$, $b_{i}(x)\in\mathcal{B}_{\lambda}^1, \deg b_i(x)\leq p^s-1$ for $0\leq i \leq t-2.$ An arbitrary left ideal contained in $\langle u \rangle$ is given as:
        \begin{align*}
            R^{t,x^{p^s}-\lambda}_{\Theta}&(u^{(t-1)-i_1}b_{i_1}(x)-u^{(t-1)-(i_1-1)}g_{i_1}(x))+\dots+\\&R^{t,x^{p^s}-\lambda}_{\Theta}( u^{(t-1)-i_n}b_{i_n}(x)-u^{(t-1)-(i_n-1)}g_{i_n}(x)),
        \end{align*} where $0\leq i_1<i_2<\dots<i_n\leq t-2$, $ \deg b_{i_j}(x)\leq p^s-1,\,g_{i_j}(x)\in R^{t,x^{p^s}-\lambda}_{\Theta}$ for $1\leq j\leq n$, and whenever $u^{(t-1)-i_j}c(x)+u^{(t-1)-(i_j-1)}g(x)\in R^{t,x^{p^s}-\lambda}_{\Theta} (u^{(t-1)-i_k}b_{i_k}(x)-u^{(t-1)-(i_k-1)}g_{i_k}(x))+\cdots+R^{t,x^{p^s}-\lambda}_{\Theta} (u^{(t-1)-i_n}b_{i_n}(x)-u^{(t-1)-(i_n-1)}g_{i_n}(x))$ for some $c(x)\in \frac{\mathbb{F}_{p^m}[x,\theta]}{\langle x^{p^s}-\lambda_0\rangle}$, $g(x)\in R^{t,x^{p^s}-\lambda}_{\Theta}$, and $j<k$, then $b_{i_j}(x)\mid_r c(x).$
        \item[(iii)] Any non-trivial ideal not contained in $\langle u \rangle$ has the form:
        $$ R^{t,x^{p^s}-\lambda}_{\Theta}(b(x)+ ug(x)) +I,$$ where $g(x)\in R^{t,x^{p^s}-\lambda}_{\Theta}$, $b(x)\in\mathcal{B}_{\lambda}^1, \deg b(x)\leq p^s-1$, $I $ is an ideal described as in Type (ii), and whenever $u^{(t-1)-i_j}c(x)+u^{(t-1)-(i_j-1)}g(x)\in R^{t,x^{p^s}-\lambda}_{\Theta} (u^{(t-1)-i_k}b_{i_k}(x)-u^{(t-1)-(i_k-1)}g_{i_k}(x))+\cdots+R^{t,x^{p^s}-\lambda}_{\Theta} (u^{(t-1)-i_n}b_{i_n}(x)-u^{(t-1)-(i_n-1)}g_{i_n}(x))$ for some $c(x)\in \frac{\mathbb{F}_{p^m}[x,\theta]}{\langle x^{p^s}-\lambda_0\rangle}$, $g(x)\in R^{t,x^{p^s}-\lambda}_{\Theta}$, and $j<k,$ then $b_{i_j}(x)\mid_r c(x).$
        \end{itemize}
\end{theorem}
Furthermore, for $k=1,$ that is for $\lambda_1\neq 0,$ we have the following result:
\begin{corollary}\label{special form p^s case}
    The left ideals of the ring $R^{t,x^{p^s}-\lambda}_{\Theta}$ are given by $a(x)+\langle x^{p^s}-\lambda\rangle $, where $a(x)\in \mathcal{B}_{\lambda}^1\cup u\mathcal{B}_{\lambda}^1\cup\dots\cup u^{t-1}\mathcal{B}_{\lambda}^1,$ if $a(x)\mid_{r}\frac{(\lambda_{0,0}x-1)^{p^s}}{a(x)}.$
\end{corollary}
\begin{proof}
    Let
    \begin{align*}
       u^{(t-1)-i}a(x)+u^{(t-1)-(i-1)}g(x)=&u^{(t-1)-i}\big(1+g(x)(\alpha_0+u\alpha_1+\dots+u^{(t-1)}\alpha_{t-1})\frac{(\lambda_{0,0}x-1)^{p^s}}{a(x)}\big)a(x)\\
        u^{(t-1)-i}a(x)+u^{(t-1)-(i-1)}g(x)=&u^{(t-1)-i}a(x)+u^{(t-1)-i}g(x)\alpha_0(u\lambda_0^{-1}\lambda_1+\dots+u^{(t-1)}\lambda_0^{-1}\lambda_{t-1})+\\&u^{(t-1)-(i-1)}g(x)\alpha_1(u\lambda_0^{-1}\lambda_1+\dots+u^{(t-1)}\lambda_0^{-1}\lambda_{t-1})+\dots+u^{(t-1)}\\&g(x)\alpha_i(u\lambda_0^{-1}\lambda_1+\dots+u^{(t-1)}\lambda_0^{-1}\lambda_{t-1})
    \end{align*}
    From the above equation, we get the following system of equations:
    \begin{align*}
        \alpha_0\lambda_0^{-1}\lambda_1&=1\\
        \alpha_0\lambda_0^{-1}\lambda_2+\alpha_1\lambda_{0}^{-1}\lambda_{1}&=0\\
        \alpha_0\lambda_0^{-1}\lambda_3+\alpha_1\lambda_0^{-1}\lambda_2+\alpha_2\lambda_0^{-1}\lambda_1&=0\\
        &\vdots\\
        \alpha_0\lambda_0^{-1}\lambda_{t-2}+\alpha_{1}\lambda_{0}^{-1}\lambda_{t-3}+\dots+\alpha_{t-2}\lambda_{0}^{-1}\lambda_{1}&=0
    \end{align*}
    The above system is consistent. Also 
    note that $\mu\big(1+g(x)(\alpha_0+u\alpha_1+\dots+u^{(t-1)}\alpha_{t-1})\frac{(\lambda_{0,0}x-1)^{p^s}}{a(x)}\big)\\=1+g(x)\lambda_0\lambda_1^{-1}\frac{(\lambda_{0,0}x-1)^{p^s}}{a(x)}.$ Then $\gcd_r(1+g(x)\lambda_0\lambda_1^{-1}\frac{(\lambda_{0,0}x-1)^{p^s}}{a(x)},(\lambda_{0,0}x-1)^{p^s})=1.$ By Lemma \ref{invertible in Fpm[x,theta]/xps-lambda_0}, $1+g(x)\lambda_0\lambda_1^{-1}\frac{(\lambda_{0,0}x-1)^{p^s}}{a(x)}$ is left invertible in $\frac{\mathbb{F}_{p^m}[x,\theta]}{\langle (\lambda_{0,0}x-1)^{p^s}\rangle}$. Consequently, by Proposition \ref{invertible in Fpm[u]/<ut>[x,Theta]/xps-lambda}, $1+g(x)(\alpha_0+u\alpha_1+\dots+u^{(t-1)}\alpha_{t-1})\frac{(\lambda_{0,0}x-1)^{p^s}}{a(x)}$ is left invertible in $R^{t,x^{p^s}-\lambda}_{\Theta}.$ Hence, $  R^{t,x^{p^s}-\lambda}_{\Theta} (u^{(t-1)-i}a(x)+u^{(t-1)-(i-1)}g(x)) = R^{t,x^{p^s}-\lambda}_{\Theta} (a(x)) .$ Similarly, one can prove that $R^{t,x^{p^s}-\lambda}_{\Theta} (b(x)+ug(x))=R^{t,x^{p^s}-\lambda}_{\Theta} (b(x)).$\hfill{$\square$}
\end{proof}
\section{Skew Constacyclic codes of length \texorpdfstring{$2p^s$}{} over \texorpdfstring{$\frac{\mathbb{F}_{p^m}[u]}{\langle u^t\rangle}$}{}}\label{sec5}
With the assumptions and notations in Section \ref{section2}, in this section, for the particular case $f(x)=x^{2p^s}-\lambda$ and an odd prime $p,$ where $\lambda=\lambda_0+u\lambda_1+\dots+u^{t-1}\lambda_{t-1}$, $\Theta(\lambda)=\lambda$, $|\Theta|\mid 2p^s$, and $\lambda_0\ne 0,$ we give the structure of ideals of $R^{t,x^{2p^s}-\lambda}_{\Theta}$. Since $\lambda_0 \in \mathbb{F}_{p^m},$ define $\tilde{\lambda}_{0,0}=\lambda_0^{p^{m-r}},$ where $s=qm+r$ and $0\leq r<m.$ Then $\tilde{\lambda}_{0,0}^{p^s}=\lambda_0.$
It is noteworthy that, from Lemma 2.3 of \cite{AkaKanSar2}, $\lambda$ is a square in $R^t$ if and only if $\lambda_0$ is a square in $\mathbb{F}_{p^m}$. We therefore consider the following two cases.

\textbf{Case 1.} When $\lambda$ is a square in $R^t.$\\
If a $\lambda$ is a square, then there exists $\tilde{\lambda}\in R^t$ such that $\lambda=\tilde{\lambda}^2$. Since $\Theta(\lambda)=\lambda$, we have $\Theta(\tilde{\lambda})=\tilde{\lambda}.$ So, $$x^{2p^s}-\lambda=(x^{p^s}-\tilde{\lambda})(x^{p^s}-\tilde{\lambda}).$$ Consequently, we have $$R^{t,x^{2p^s}-\lambda}_{\Theta}\cong R^t[x,\Theta]/\langle x^{p^s}-\tilde{\lambda}\rangle\oplus R^t[x,\Theta]/\langle x^{p^s}+\tilde{\lambda}\rangle.$$
That is, skew constacyclic codes of length $2p^s$ over $R^t$ can be identified as a direct product of skew constacyclic codes of length $p^s$ over $R^t.$

\textbf{Case 2.} When $\lambda$ is not a square in $R^t.$

For the rest of this section we assume that $\lambda$ is not a square, then from Lemma 2.3 of \cite{AkaKanSar2}, $\lambda_0$ is not a square in $\mathbb{F}_{p^m}.$ 

We have the following theorem for the central elements of $R^{t,x^{2p^s}-\lambda}_{\Theta}$.
\begin{theorem} [\cite{Jitman}, Proposition 2.3]
Let $g(x),h(x)\in R^{t,x^{2p^s}-\lambda}_{\Theta}.$ If $g(x)h(x)$ is in the center of the ring $R^{t,x^{2p^s}-\lambda}_{\Theta}$, then $g(x)h(x)=h(x)g(x).$
\end{theorem}
    \begin{lemma}\label{relation between2}
        Let $k\geq 1$ be the smallest integer such that $\lambda_k\neq 0.$ Then in $R^{t,x^{2p^s}-\lambda}_{\Theta},$ we have $\langle  (x^2-\tilde{\lambda}_{0,0})^{p^s}\rangle =\langle u^k\rangle.$ 
    \end{lemma}
    \begin{proof}
        In  $R^{t,x^{2p^s}-\lambda}_{\Theta},$ we have $(x^2-\tilde{\lambda}_{0,0})^{p^s}=\lambda_0^{-1}\lambda -1=u^k(\lambda_0^{-1}\lambda_k+\lambda_0^{-1}\lambda_{k+1}u+\dots+\lambda_0^{-1}\lambda_{t-1}u^{t-1-k}).$ Hence, $\langle (x^2-\tilde{\lambda}_{0,0})^{p^s} \rangle=\langle u^k\rangle.$ \hfill $\square$
    \end{proof}
    Hence $x^2-\tilde{\lambda}_{0,0}$ is a nilpotent in the ring $R^{t,x^{2p^s}-\lambda}_{\Theta}$ having nilpotency index $\lceil t/k \rceil.$ Due to Theorem 2.2 in Section \ref{section2}, we can give the following representation of every element of $R^{t,x^{2p^s}-\lambda}_{\Theta}.$ 
    \begin{remark}\label{Representation2}
    \begin{itemize}
        \item[(i)] An arbitrary element $c(x)$ of $R^{t,x^{2p^s}-\lambda}_{\Theta}$ can be uniquely written as 
        \begin{align*}\label{form of elt in t2}
             c(x)=\underset{i=0}{\overset{p^s-1}{\sum}}(a_{0,i}x+b_{0,i})(x^2-\tilde{\lambda}_{0,0})^i+&u\underset{i=0}{\overset{p^s-1}{\sum}}(a_{1,i}x+b_{1,i})(x^2-\tilde{\lambda}_{0,0})^i+\dots+\\&u^{t-1}\underset{i=0}{\overset{p^s-1}{\sum}}(a_{t-1,i}x+b_{t-1,i})(x^2-\tilde{\lambda}_{0,0})^i
        \end{align*}
        where $a_{j,i},b_{j,i}\in \mathbb{F}_{p^m}$ for $0\leq j \leq t-1,\,0\leq i\leq p^s-1.$
    \end{itemize}
    \end{remark}
    It is noteworthy that in a non-commutative ring, the sum of a unit and a nilpotent element may not be a unit. Define $$\mathcal{B}_{\lambda}^2=\{a(x)\in \frac{\mathbb{F}_{p^m}[x,\theta]}{\langle (x^2-\tilde{\lambda}_{0,0})^{p^s}\rangle}\,| \,a(x)\mid_r(x^2-\tilde{\lambda}_{0,0})^{p^s}\}.$$
Due to Proposition 4.1 and Proposition 4.2 in \cite{hesari2023skew}, we have the following results.
\begin{proposition}[Proposition 4.1, \cite{hesari2023skew}]
    The ring $\frac{\mathbb{F}_{p^m}[x,\theta]}{\langle x^{2p^s}-\lambda_0\rangle}$ is a principal ideal ring. In this ring the left ideal are generated by $a(x)+\langle (x^2-\tilde{\lambda}_{0,0})^{p^s}\rangle,$ where $a(x)$ is right divisor of $(x^2-\tilde{\lambda}_{0,0})^{p^s}$ in $\mathbb{F}_{p^m}[x,\theta].$
\end{proposition}
\begin{proposition}[Proposition 4.2, \cite{hesari2023skew}]
    Let $g(x)$ be a non-zero polynomial in $\frac{\mathbb{F}_{p^m}[x,\theta]}{\langle x^{2p^s}-\lambda_0\rangle}$. Then $g(x)$ is a left invertible element in $\frac{\mathbb{F}_{p^m}[x,\theta]}{\langle x^{2p^s}-\lambda_0\rangle}$ if and only if $\gcd_r(g(x),(x^2-\tilde{\lambda}_{0,0})^{p^s})=1.$
\end{proposition}
\begin{proposition}
    Let $h(x)$ be a non-zero element in $R^{t,x^{2p^s}-\lambda}_{\Theta}.$ Then $h(x)$ is left invertible in $R^{t,x^{2p^s}-\lambda}_{\Theta}$ if and only if $\mu(h(x))$ is left invertible in $\frac{\mathbb{F}_{p^m}[x,\theta]}{\langle x^{2p^s}-\lambda_0\rangle}.$
\end{proposition}
\begin{proof}
    The proof follows along lines similar to those of the Proposition \ref{invertible in Fpm[u]/<ut>[x,Theta]/xps-lambda}.\hfill{$\square$}
\end{proof}
Due to the above proposition, the left ideals of $R^{t,x^{2p^s}-\lambda}_{\Theta}$ are given by the following theorem, and the proof of this theorem follows on similar lines as in Theorem \ref{nmain theorem}.
\begin{theorem}\label{2main theorem}
    The left ideals of the ring $R^{t,x^{2p^s}-\lambda}_{\Theta}$, where $\lambda$ is not a square in $R^t$, are of the following form.
    \begin{itemize}
        \item [(i)] Trivial ideals $\langle 0\rangle,$ $\langle 1 \rangle.$
        \item [(ii)] Any generator of a non-trivial left ideal contained in $\langle u\rangle $ has the form:
        $$u^{(t-1)-i}b_i(x)-u^{(t-1)-(i-1)}g_i(x),$$
        where $g_i(x)\in R^{t,x^{2p^s}-\lambda}_{\Theta}$, $b_{i}(x)\in\mathcal{B}_{\lambda}^2, \deg b_i(x)\leq 2p^s-1$ for $0\leq i \leq t-2.$ An arbitrary left ideal contained in $\langle u \rangle$ is given as:
        \begin{align*}
            R^{t,x^{2p^s}-\lambda}_{\Theta}&(u^{(t-1)-i_1}b_{i_1}(x)-u^{(t-1)-(i_1-1)}g_{i_1}(x))+\dots+\\&R^{t,x^{2p^s}-\lambda}_{\Theta}( u^{(t-1)-i_n}b_{i_n}(x)-u^{(t-1)-(i_n-1)}g_{i_n}(x)),
        \end{align*} where $0\leq i_1<i_2<\dots<i_n\leq t-2$, $ \deg b_{i_j}(x)\leq 2p^s-1,\,g_{i_j}(x)\in R^{t,x^{2p^s}-\lambda}_{\Theta}$ for $1\leq j\leq n$, and whenever $u^{(t-1)-i_j}c(x)+u^{(t-1)-(i_j-1)}g(x)\in R^{t,x^{2p^s}-\lambda}_{\Theta} (u^{(t-1)-i_k}b_{i_k}(x)-u^{(t-1)-(i_k-1)}g_{i_k}(x))+\cdots+R^{t,x^{2p^s}-\lambda}_{\Theta} (u^{(t-1)-i_n}b_{i_n}(x)-u^{(t-1)-(i_n-1)}g_{i_n}(x))$ for some $c(x)\in \frac{\mathbb{F}_{p^m}[x,\theta]}{\langle x^{2p^s}-\lambda_0\rangle}$, $g(x)\in R^{t,x^{2p^s}-\lambda}_{\Theta}$, and $j<k,$ then $b_{i_j}(x)\mid_r c(x).$
        \item[(iii)] Any non-trivial ideal not contained in $\langle u \rangle$ has the form:
        $$ R^{t,x^{2p^s}-\lambda}_{\Theta}(b(x)+ ug(x)) +I,$$ where $g(x)\in R^{t,x^{2p^s}-\lambda}_{\Theta}$, $b(x)\in\mathcal{B}_{\lambda}^2, \deg b(x)\leq 2p^s-1$, $I $ is an ideal described as in Type (ii), and whenever $u^{(t-1)-i_j}c(x)+u^{(t-1)-(i_j-1)}g(x)\in R^{t,x^{2p^s}-\lambda}_{\Theta} (u^{(t-1)-i_k}b_{i_k}(x)-u^{(t-1)-(i_k-1)}g_{i_k}(x))+\cdots+R^{t,x^{2p^s}-\lambda}_{\Theta} (u^{(t-1)-i_n}b_{i_n}(x)-u^{(t-1)-(i_n-1)}g_{i_n}(x))$ for some $c(x)\in \frac{\mathbb{F}_{p^m}[x,\theta]}{\langle x^{2p^s}-\lambda_0\rangle}$, $g(x)\in R^{t,x^{2p^s}-\lambda}_{\Theta}$, and $j<K,$ then $b_{i_j}(x)\mid_r c(x).$
        \end{itemize}
\end{theorem}
Furthermore, for $k=1,$ that is for $\lambda_1\neq 0,$ we have the following result:
\begin{corollary}
    The left ideals of the ring $R^{t,x^{2p^s}-\lambda}_{\Theta}$ are given by $a(x)+\langle x^{2p^s}-\lambda_0\rangle $, where $a(x)\in \mathcal{B}_{\lambda}^2\cup u\mathcal{B}_{\lambda}^2\cup\dots\cup u^{t-1}\mathcal{B}_{\lambda}^2,$ if $a(x)\mid_{r}\frac{(x^2-\lambda_{0,0})^{p^s}}{a(x)}.$
\end{corollary}
\begin{proof}
    The proof follows along lines similar to those of the Corollary \ref{special form p^s case}. \hfill{$\square$}
\end{proof}
\section{Applications}\label{sec6}
\subsection{Skew constacyclic codes of length \texorpdfstring{$p^s$}{} over \texorpdfstring{$\frac{\mathbb{F}_{p^m}[u]}{\langle u^3\rangle}$}{}
}
In this section, we analyze the special case $t=3$, that is, $R^{3,x^{p^s}-\lambda}_{\Theta}$. The case for  $t=3$ has been considered in  \cite{hesari2025skew}, but in the description of the ideals of some types, there is a slight gap.
We begin by completing the description of ideals for $t=3$. That is, in the following theorem, we give a complete description of the ideals of $R^{3,x^{p^s}-\lambda}_{\Theta}:=\frac{\mathbb{F}_{p^m}[u]}{\langle u^3\rangle} [x,\Theta]/\langle x^{p^s}-\lambda\rangle,$ where $\lambda=\lambda_0+u\lambda_1+u^2\lambda_2$ and $\lambda_0\in\mathbb{F}_{p^m}\setminus \{0\}$. Recall that $\mathcal{B}_{\lambda}^1:=\{a(x)\in\frac{\mathbb{F}_{p^m}[x,\theta]}{\langle (\lambda_{0,0}x-1)^{p^s}\rangle}\,:\,a(x)\mid_{r}(\lambda_{0,0}x-1)^{p^s})\}$. Also for any $a(x)\in \mathcal{B}_{\lambda}^1,$ we have $(\lambda_{0,0}x-1)^{p^s}=\widetilde{a(x)}a(x).$

\begin{theorem}\label{t=3 case}
    Left ideals of $R^{3,x^{p^s}-\lambda}_{\Theta}$ are of the following form:
    \begin{enumerate}
        \item {\label{Type 1}} $C_0=0$,\,$C_1=R^{3,x^{p^s}-\lambda}_{\Theta}.$
        \item {\label{Type 2}}$C_2=R^{3,x^{p^s}-\lambda}_{\Theta}u^2b(x)$, where $b(x)\in \mathcal{B}_{\lambda}^1$ and $\deg b(x)\leq p^s-1$.
         \item {\label{Type 3}} $C_3=R^{3,x^{p^s}-\lambda}_{\Theta}\big( ub(x)+u^2g(x)\big)$, where $ b(x)\in \mathcal{B}_{\lambda}^1,\deg b(x)\leq p^s-1$, $g(x)\in \frac{\mathbb{F}_{p^m}[x,\theta]}{\langle (\lambda_{0,0}x-1)^{p^s}\rangle},\, \\\deg g(x)<\deg b(x)$. Moreover, $g(x)$ with these conditions is unique.
        \item {\label{Type 4}}$C_4=R^{3,x^{p^s}-\lambda}_{\Theta} \big(ub_{1}(x)+u^2 g(x)\big)+R^{3,x^{p^s}-\lambda}_{\Theta} u^2b_{2}(x),$ where $b_i(x)\in \mathcal{B}_{\lambda}^1$ and $\deg b_i(x)\leq p^s-1$ for $1\leq i\leq 2$, $ g(x)\in \frac{\mathbb{F}_{p^m}[x,\theta]}{\langle (\lambda_{0,0}x-1)^{p^s}\rangle}$, $\deg g(x)<\deg b_2(x),$ and if $b(x)\in \frac{\mathbb{F}_{p^m}[x,\theta]}{\langle (\lambda_{0,0}x-1)^{p^s}\rangle}$  be such that $u^2b(x)\in R^{3,x^{p^s}-\lambda}_{\Theta} \big(ub_{1}(x)+u^2 g(x)\big)$ then $b_2(x)\mid_rb(x)$ in $\frac{\mathbb{F}_{p^m}[x,\theta]}{\langle (\lambda_{0,0}x-1)^{p^s}\rangle}.$ Moreover, $g(x)$ with these conditions is unique.
       \item {\label{Type 5}}$C_5=R^{3,x^{p^s}-\lambda}_{\Theta}\big(  b(x)+ug_1(x)+u^2g_2(x)\big),$ where $ b(x)\in \mathcal{B}_{\lambda}^1$ and $\deg b(x)\leq p^s-1,\, g_j(x)\in \frac{\mathbb{F}_{p^m}[x,\theta]}{\langle (\lambda_{0,0}x-1)^{p^s}\rangle}$ for $1\leq j\leq 2,\,\max\{\deg g_1(x),\,\deg g_2(x)\}<\deg b(x).$ Moreover, $g_i(x)$ with these conditions are unique for $i=1,2.$
        \item {\label{Type 6}} $C_6=R^{3,x^{p^s}-\lambda}_{\Theta}\big(  b_1(x)+ug_1(x)+u^2g_2(x)\big)+R^{3,x^{p^s}-\lambda}_{\Theta}u^2b_2(x),$ where $b_i(x)\in \mathcal{B}_{\lambda}^1$ and $\deg b_i(x)\leq p^s-1$ for $i=1,2$, $g_i(x)\in \frac{\mathbb{F}_{p^m}[x,\theta]}{\langle (\lambda_{0,0}x-1)^{p^s}\rangle}$ for $i=1,2$, $\deg g_1(x)<\deg b_1(x),\\\deg g_2(x)<\deg b_2(x)$, and if $b(x)\in \frac{\mathbb{F}_{p^m}[x,\theta]}{\langle (\lambda_{0,0}x-1)^{p^s}\rangle}$  be such that $u^2b(x)\in R^{3,x^{p^s}-\lambda}_{\Theta}\big(  b_1(x)+ug_1(x)+u^2g_2(x)\big)$ then $b_2(x)\mid_r b(x)$ in $\frac{\mathbb{F}_{p^m}[x,\theta]}{\langle (\lambda_{0,0}x-1)^{p^s}\rangle}.$ Moreover, $g_i(x)$ with these assumptions are unique for $i=1,2.$
        \item {\label{Type 7}}$C_7=R^{3,x^{p^s}-\lambda}_{\Theta}\big(b_1(x)+ug_1(x)+u^2g_2(x)\big)+R^{3,x^{p^s}-\lambda}_{\Theta}\big(ub_2(x)+u^2g_3(x)\big), $  where $ b_i(x)\in \mathcal{B}_{\lambda}^1 $ and $\deg b_i(x)\leq p^s-1$ for $1\leq i\leq 2$, $ g_j(x)\in \frac{\mathbb{F}_{p^m}[x,\theta]}{\langle (\lambda_{0,0}x-1)^{p^s}\rangle}$ for $1\leq j\leq 3$, and if $b(x), g(x)\in \frac{\mathbb{F}_{p^m}[x,\theta]}{\langle (\lambda_{0,0}x-1)^{p^s}\rangle}$ be such that $ub(x)+u^2g(x)\in R^{3,x^{p^s}-\lambda}_{\Theta}\big(b_1(x)+ug_1(x)+u^2g_2(x)\big)$ then $b_2(x)\mid_r b(x)$ in $\frac{\mathbb{F}_{p^m}[x,\theta]}{\langle (\lambda_{0,0}x-1)^{p^s}\rangle}.$ Moreover, $g_i(x)$ with these conditions are unique for $i=1,2,$ and $3.$
        
        \item {\label{Type 8}} $C_8=R^{3,x^{p^s}-\lambda}_{\Theta}\big(b_1(x)+ug_1(x)+u^2g_2(x)\big)+R^{3,x^{p^s}-\lambda}_{\Theta}\big(ub_2(x)+u^2g_3(x)\big)+R^{3,x^{p^s}-\lambda}_{\Theta}u^2b_3(x), $  where $ b_i(x)\in \mathcal{B}_{\lambda}^1 $ and $\deg b_i(x)\leq p^s-1$ for $1\leq i\leq 3$, $ g_j(x)\in \frac{\mathbb{F}_{p^m}[x,\theta]}{\langle (\lambda_{0,0}x-1)^{p^s}\rangle}$ for $1\leq j\leq 3$, if $b(x)\in \frac{\mathbb{F}_{p^m}[x,\theta]}{\langle (\lambda_{0,0}x-1)^{p^s}\rangle}$ be such that $u^2b(x)\in R^{3,x^{p^s}-\lambda}_{\Theta}\big(b_1(x)+ug_1(x)+u^2g_2(x)\big)+R^{3,x^{p^s}-\lambda}_{\Theta}\big(ub_2(x)+u^2g_3(x)\big)$ then $b_3(x)\mid_r b(x),$ and if $b(x), g(x)\in \frac{\mathbb{F}_{p^m}[x,\theta]}{\langle (\lambda_{0,0}x-1)^{p^s}\rangle}$ be such that $ub(x)+u^2g(x)\in R^{3,x^{p^s}-\lambda}_{\Theta}\big(b_1(x)+ug_1(x)+u^2g_2(x)\big)$ then $b_2(x)\mid_r b(x)$ in $\frac{\mathbb{F}_{p^m}[x,\theta]}{\langle (\lambda_{0,0}x-1)^{p^s}\rangle}.$ Moreover, $g_i(x)$ with these conditions are unique for $i=1,2,$ and $3.$
    \end{enumerate}
\end{theorem}
We remark that, in the existing literature \cite{hesari2025skew}, the descriptions of left ideals in Theorem \ref{t=3 case} of Types $\ref{Type 6}$, $\ref{Type 7}$, and $\ref{Type 8}$ are not complete, which may be due to an oversight. In particular, for Type $\ref{Type 6}$, a certain necessary divisibility condition is not included, namely,
``if $b(x)\in \frac{\mathbb{F}_{p^m}[x,\theta]}{\langle (\lambda_{0,0}x-1)^{p^s}\rangle}$  be such that $u^2b(x)\in R^{3,x^{p^s}-\lambda}_{\Theta}\big(  b_1(x)+ug_1(x)+u^2g_2(x)\big)$ then $b_2(x)\mid_r b(x).$"
Similarly, the corresponding divisibility conditions are missing in the descriptions of Types $7$ and $8$. The omission of these conditions leads to intersection in different types and incorrect determination of the associated torsion codes.
Due to this addition in the description of the ideals, we get the torsions in Theorem 4.1 \cite{hesari2025skew} correctly, and we describe each case (Type 6, Type 7, and Type 8) in the following corollaries.

\begin{corollary}\label{cor1forps}
If $C_6$ is as described in Theorem \ref{t=3 case}.\ref{Type 6}, then $\Tor_2(C_6)=\frac{\mathbb{F}_{p^m}[x,\theta]}{\langle (\lambda_{0,0}x-1)^{p^s}\rangle} b_2(x).$
\end{corollary}
\begin{proof}
It is clear that $b_2(x)\in \Tor_2(C_6)$ which implies $\frac{\mathbb{F}_{p^m}[x,\theta]}{\langle (\lambda_{0,0}x-1)^{p^s}\rangle} b_2(x)\subset \Tor_2(C_6).$ Conversely, let $c(x)\in R^{3,x^{p^s}-\lambda}_{\Theta}$ be an arbitrary element such that $\mu(c(x)=c_0(x)+uc_1(x)+u^2c_2(x))\in \Tor_2(C_6)$. Then $u^2c(x)\in C_6.$ That is, $u^2c_0(x)\in C_6$, and hence, there exist $a(x)=a_0(x)+ua_1(x)+u^2a_2(x)$ and $m(x)=m_0(x)+um_1(x)+u^2m_2(x)\in R^{3,x^{p^s}-\lambda}_{\Theta}$ such that $u^2c_0(x)=a(x)(b_1(x)+ug_1(x)+u^2g_2(x))+m(x)(u^2b_2(x))\implies u^2c_0(x)-m_0(x)(u^2b_2(x))=a(x)(b_1(x)+ug_1(x)+u^2g_2(x))$ and hence $b_2(x)\mid_r c_0(x)-m_0(x)b_2(x)$, consequently, $b_2(x)\mid_r  c_0(x)$.
\end{proof}

\begin{corollary}\label{cor2forps}
If $C_7$ is as described in Theorem \ref{t=3 case}.\ref{Type 7}, then $\Tor_1(C_7)=\frac{\mathbb{F}_{p^m}[x,\theta]}{\langle (\lambda_{0,0}x-1)^{p^s}\rangle} b_2(x).$
\end{corollary}
\begin{proof}
    It is clear that $b_2(x)\in \Tor_1(C_7)$ which implies $\frac{\mathbb{F}_{p^m}[x,\theta]}{\langle (\lambda_{0,0}x-1)^{p^s}\rangle} b_2(x)\subset \Tor_1(C_7).$ Conversely, let $c(x)\in R^{3,x^{p^s}-\lambda}_{\Theta}$ be an arbitrary element such that $\mu(c(x)=c_0(x)+uc_1(x)+u^2c_2(x))\in \Tor_1(C_7)$. Then $uc(x)\in C_7.$ That is, there exist $a(x)=a_0(x)+ua_1(x)+u^2a_2(x), m(x)=m_0(x)+um_1(x)+u^2m_2(x)\in R^{3,x^{p^s}-\lambda}_{\Theta}$ such that $uc_0(x)+u^2c_1(x)=a(x)(b_1(x)+ug_1(x)+u^2g_2(x))+m(x)(ub_2(x)+u^2g_3(x))\implies u(c_0(x)-m_0(x)b_2(x))+u^2(c_1(x)-m_0(x)g_3(x)-m_1(x)b_2(x))=a(x)(b_1(x)+ug_1(x)+u^2g_2(x))$ and hence $b_2(x)\mid_r c_0(x)-m_0(x)b_2(x)$, consequently, $b_2(x)\mid_r  c_0(x)$.
\end{proof}
\begin{corollary}\label{cor3forps}
If $C_8$ is as described in Theorem \ref{t=3 case}.\ref{Type 8}, then $\Tor_1(C_8)=\frac{\mathbb{F}_{p^m}[x,\theta]}{\langle (\lambda_{0,0}x-1)^{p^s}\rangle} b_2(x)$ and $\Tor_2(C_8)=\frac{\mathbb{F}_{p^m}[x,\theta]}{\langle (\lambda_{0,0}x-1)^{p^s}\rangle} b_3(x).$
\end{corollary}
\begin{proof}
      It is clear that $b_2(x)\in \Tor_1(C_8)$ which implies $\frac{\mathbb{F}_{p^m}[x,\theta]}{\langle (\lambda_{0,0}x-1)^{p^s}\rangle} b_2(x)\subset \Tor_1(C_8).$ Conversely, let $c(x)\in R^{3,x^{p^s}-\lambda}_{\Theta}$ be an arbitrary element such that $\mu(c(x)=c_0(x)+uc_1(x)+u^2c_2(x))\in \Tor_1(C_8)$. Then $uc(x)\in C_8.$ That is, there exist $a(x)=a_0(x)+ua_1(x)+u^2a_2(x), m(x)=m_0(x)+um_1(x)+u^2m_2(x), $ and $n(x)=n_0(x)+un_1(x)+u^2n_2(x)\in R^{3,x^{p^s}-\lambda}_{\Theta}$ such that $uc_0(x)+u^2c_1(x)=a(x)(b_1(x)+ug_1(x)+u^2g_2(x))+m(x)(ub_2(x)+u^2g_3(x))+n(x)(u^2b_3(x))\implies u(c_0(x)-m_0(x)b_2(x))+u^2(c_1(x)-m_0(x)g_3(x)-m_1(x)b_2(x)-n_0(x)b_3(x))=a(x)(b_1(x)+ug_1(x)+u^2g_2(x))$ and hence $b_2(x)\mid_r c_0(x)-m_0(x)b_2(x)$, consequently, $b_2(x)\mid_r  c_0(x)$.
      
      For $\Tor_2(C_8),$ it is clear that $b_3(x)\in \Tor_2(C_8)$ which implies $\frac{\mathbb{F}_{p^m}[x,\theta]}{\langle (\lambda_{0,0}x-1)^{p^s}\rangle} b_3(x)\subset \Tor_2(C_8).$ Conversely, let $c(x)\in R^{3,x^{p^s}-\lambda}_{\Theta}$ be an arbitrary element such that $\mu(c(x)=c_0(x)+uc_1(x)+u^2c_2(x))\in \Tor_2(C_8)$. Then $u^2c(x)\in C_8.$ That is, there exist $a(x)=a_0(x)+ua_1(x)+u^2a_2(x), m(x)=m_0(x)+um_1(x)+u^2m_2(x), $ and $n(x)=n_0(x)+un_1(x)+u^2n_2(x)\in R^{3,x^{p^s}-\lambda}_{\Theta}$ such that $u^2c_0(x)=a(x)(b_1(x)+ug_1(x)+u^2g_2(x))+m(x)(ub_2(x)+u^2g_3(x))+n(x)(u^2b_3(x))\implies u^2(c_0(x)-n_0(x)b_3(x))=a(x)(b_1(x)+ug_1(x)+u^2g_2(x))+m(x)(ub_2(x)+u^2g_3(x))$ and hence $b_3(x)\mid_r c_0(x)-n_0(x)b_3(x)$, consequently, $b_3(x)\mid_r  c_0(x)$. 
\end{proof}
\begin{remark}
Corollaries \ref{cor1forps}, \ref{cor2forps}, and \ref{cor3forps} demonstrate that the additional conditions (missing in \cite{hesari2025skew}) imposed in the characterization of ideals in Theorem \ref{t=3 case} are necessary for these families to be disjoint. In fact, to compute the torsion, the descriptions of ideals must yield disjoint families of ideals of the ring.  
\end{remark}
The following example shows that the description of ideals of Type 6 in Theorem 3.15 in \cite{hesari2025skew} for the commutative case is not good enough to compute the torsion.
\begin{example}
    Consider the quotient ring $\frac{R^2[x]}{\langle (x-1)^4\rangle}.$ Assume that $\Theta$ is the identity map on $R^2$. Let  $C$ be an ideal of the form $\langle (x-1)^3+u(x-1)\rangle +\langle u^2(x-1)^2\rangle$. Consider $b_1(x)=(x-1)^3,$ $b_2(x)=(x-1)^2$ and so on, it satisfies all the conditions given in  Type 6 of Theorem 3.15 in \cite{hesari2025skew}. However, $\Tor_2(C)$ is not $\langle (x-1)^2\rangle$, it is equal to $\langle x-1\rangle,$ as $
    (u-(x-1)^2)((x-1)^3+u(x-1))=u^2(x-1)\in C.$  In fact, the ideal $C$ is equal to $\langle (x-1)^3+u(x-1)\rangle,$ and hence Type 6 in Theorem 3.15 of \cite{hesari2025skew} reduces to Type 5. But if we consider the additional condition of Type \ref{Type 6} in Theorem \ref{t=3 case}, no ideals will belong to two different types. 
\end{example}
\subsection{Skew constacyclic codes of length \texorpdfstring{$2p^s$}{} over \texorpdfstring{$\frac{\mathbb{F}_{p^m}[u]}{\langle u^2\rangle}$}{}
}
In this section, we analyze the special case $t=2$. The case for  $t=2$ has been considered in  \cite{hesari2023skew}, but in the description of the ideal of Type 4 in Theorem 5.2 (\cite{hesari2023skew}), one condition is missing. We begin by completing the description of ideals for $t=2$. That is, in the following theorem, we give a complete description of the ideals of $R^{2,x^{2p^s}-\lambda}_{\Theta}:=\frac{R^2 [x,\Theta]}{\langle x^{2p^s}-\lambda\rangle},$ where $\lambda=\lambda_0+u\lambda_1$ and $\lambda_0\in\mathbb{F}_{p^m}\setminus \{0\}$. Recall that $\lambda$ is a square if and only if $\lambda_0$ is not a square in $\mathbb{F}_{p^m}$ [Lemma 2.3 of \cite{AkaKanSar2}]. If $\lambda$ is a square, we are done due to the Chinese Remainder Theorem. Let us assume that $\lambda$ is not a square, which implies that $\lambda_0$ is not a square. Also there exists $\tilde{\lambda}_{0,0}\in\mathbb{F}_{p^m}$ such that $\tilde{\lambda}_{0,0}\in\mathbb{F}_{p^m}^{p^s}=\tilde{\lambda}_{0,0}\in\mathbb{F}_{p^m}.$ Recall that $\mathcal{B}_{\lambda}^2:=\{a(x)\in\frac{\mathbb{F}_{p^m}[x,\theta]}{\langle (x^2-\tilde{\lambda}_{0,0})^{p^s}\rangle}\,:\,a(x)\mid_{r}(x^2-\tilde{\lambda}_{0,0})^{p^s}\}$. Also for any $a(x)\in \mathcal{B}_{\lambda}^2,$ we have $(x^2-\tilde{\lambda}_{0,0})^{p^s}=\widetilde{a(x)}a(x).$
\begin{theorem}\label{t=2 case}
    Left ideals of $R^{2,x^{2p^s}-\lambda}_{\Theta}$ are of the following form:
    \begin{enumerate}
        \item {\label{Type 1'}} $C_0=0$,\,$C_1=R^{2,x^{2p^s}-\lambda}_{\Theta}.$
        \item {\label{Type 2'}}$C_2=R^{2,x^{2p^s}-\lambda}_{\Theta}ub(x)$, where $b(x)\in \mathcal{B}_{\lambda}^2$ and $\deg b(x)\leq 2p^s-1$.
         \item {\label{Type 3'}} $C_3=R^{2,x^{2p^s}-\lambda}_{\Theta}\big( b(x)+ug(x)\big)$, where $ b(x)\in \mathcal{B}_{\lambda}^2,\deg b(x)\leq 2p^s-1$, $g(x)\in \frac{\mathbb{F}_{p^m}[x,\theta]}{\langle (x^2-\tilde{\lambda}_{0,0})^{p^s}\rangle},\, \\\deg g(x)<\deg b(x)$. Moreover, $g(x)$ with these conditions is unique.
        \item {\label{Type 4'}}$C_4=R^{2,x^{2p^s}-\lambda}_{\Theta} \big(b_{1}(x)+ug(x)\big)+R^{2,x^{2p^s}-\lambda}_{\Theta}ub_2(x),$ where $b_i(x)\in \mathcal{B}_{\lambda}^2$ and $\deg b_i(x)\leq 2p^s-1$ for $1\leq i\leq 2$, $ g(x)\in \frac{\mathbb{F}_{p^m}[x,\theta]}{\langle (x^2-\tilde{\lambda}_{0,0})^{p^s}\rangle},\,\deg g(x)<\deg b_2(x),$ and if $b(x)\in \frac{\mathbb{F}_{p^m}[x,\theta]}{\langle (x^2-\tilde{\lambda}_{0,0})^{p^s}\rangle}$  be such that $ub(x)\in R^{2,x^{2p^s}-\lambda}_{\Theta}\big(  b_1(x)+ug(x)\big)$ then $b_2(x)\mid_r b(x)$ in $\frac{\mathbb{F}_{p^m}[x,\theta]}{\langle (x^2-\tilde{\lambda}_{0,0})^{p^s}\rangle}.$ Moreover, $g(x)$ with these conditions is unique.
\end{enumerate}
\end{theorem}
Due to ``if $b(x)\in \frac{\mathbb{F}_{p^m}[x,\theta]}{\langle (x^2-\tilde{\lambda}_{0,0})^{p^s}\rangle}$  be such that $u^2b(x)\in R^{2,x^{2p^s}-\lambda}_{\Theta}\big(  b_1(x)+ug(x)\big)$ then $b_2(x)\mid_r b(x)$ in $\frac{\mathbb{F}_{p^m}[x,\theta]}{\langle (x^2-\tilde{\lambda}_{0,0})^{p^s}\rangle}$" in the description of the ideal Type 4 above, which was missing in Theorem 5.2 \cite{hesari2023skew}, we get the $\Tor_1(C_4):=\mu\{a(x)\in R^{2,x^{2p^s}-\lambda}_{\Theta}\,|\,a(x)u\in C_4\}$ correctly. 
\begin{corollary}
If $C_4$ is as described in Theorem \ref{t=2 case}.\ref{Type 4'}, then $\Tor_1(C_4)=\frac{\mathbb{F}_{p^m}[x,\theta]}{\langle (x^2-\tilde{\lambda}_{0,0})^{p^s}\rangle} b_2(x)$.
\end{corollary}
\begin{proof}
    It is clear that $b_2(x)\in \Tor_1(C_4)$ which implies $\frac{\mathbb{F}_{p^m}[x,\theta]}{\langle (x^2-\tilde{\lambda}_{0,0})^{p^s}\rangle}  b_2(x)\subset \Tor_1(C_4).$ Conversely, let $c(x)\in R^{2,x^{2p^s}-\lambda}_{\Theta}$ be an arbitrary element such that $\mu(c(x)=c_0(x)+uc_1(x)+u^2c_2(x))\in \Tor_1(C_4)$. Then $uc(x)\in C_4.$ That is, there exist $a(x)=a_0(x)+ua_1(x)$ and $m(x)=m_0(x)+um_1(x)\in R^{2,x^{2p^s}-\lambda}_{\Theta}$ such that $uc_0(x)=a(x)(b_1(x)+ug(x))+m(x)(ub_2(x))\implies u(c_0(x)-m_0(x)b_2(x))=a(x)(b_1(x)+ug(x))$ and hence $b_2(x)\mid_r c_0(x)-m_0(x)b_2(x)$, consequently, $b_2(x)\mid_r  c_0(x)$.
      
\end{proof}
\begin{remark}
The above corollary demonstrates that the additional condition (missing in Theorem 5.2 \cite{hesari2023skew}) imposed in the characterization of ideals in Theorem \ref{t=2 case} is necessary for these families to be disjoint.  A similar situation is seen in the description of left ideals in \cite{pathak2025skew} in Theorem 4.15, which can be fixed in a similar manner.
\end{remark}
\section{Conclusion}\label{sec7}
In this work, we determine the complete structure of the left ideals of the ring  $\frac{R^t[x,\Theta]}
{\langle f(x) \rangle}$, where $f(x) \in 
R^t[x,\Theta]$ is in the center and $\Theta$ is a special kind of automorphism of $R^t$. Further, we have obtained a refined characterization when $f(x)=x^{p^s}-\lambda$ with $\lambda_1 \neq 0$, and when $f(x)=x^{2p^s}-\lambda$ with $\lambda_1 \neq 0$. As an application of the general theory, in particular cases, we listed all types of skew constacyclic codes of length $p^s$ and $2p^s$ over $R^3$. It is noteworthy that we have added necessary conditions in the description of the ideals so that there is no intersection among the types of the ideals. In this setting, we further computed the associated $i$-th torsion codes. The results in this article extend the work done so far by the authors in the skew setting of ideals. For future investigations, one can explore duality, distance, and enumeration of the corresponding codes.




\textbf{Conflict of Interest.} Both authors declare that they have no conflict of interest.
\section*{Acknowledgements}
The first author would like to acknowledge PMRF (PMRF Id: 1403187) for its financial support.
\nocite{dinh2012repeated, dinh2020constacyclic, dinh2013structure,dinh2017constacyclic, dinh2013repeated, dinh2024constacyclic, dinh2024hamming, consta8ps}
\bibliographystyle{abbrv}
\bibliography{sample}

\end{document}